\pdfoutput=1
\RequirePackage{ifpdf}
\ifpdf 
\documentclass[pdftex]{sigma}
\else
\documentclass{sigma}
\fi

\usepackage[all]{xy}

\newcommand\unit{\hbox{\rm 1\kern-2.8truept l}}

\numberwithin{equation}{section}

\newtheorem{Theorem}{Theorem}[section]

\newtheorem{Lemma}[Theorem]{Lemma}
\newtheorem{Proposition}[Theorem]{Proposition}
 { \theoremstyle{definition}

\newtheorem{Remark}[Theorem]{Remark} }

\begin{document}

\newcommand{\arXivNumber}{2202.02196}

\renewcommand{\thefootnote}{}

\renewcommand{\PaperNumber}{035}

\FirstPageHeading

\ShortArticleName{The Generalized Fibonacci Oscillator as an Open Quantum System}

\ArticleName{The Generalized Fibonacci Oscillator\\ as an Open Quantum System\footnote{This paper is a~contribution to the Special Issue on Non-Commutative Algebra, Probability and Analysis in Action. The~full collection is available at \href{https://www.emis.de/journals/SIGMA/non-commutative-probability.html}{https://www.emis.de/journals/SIGMA/non-commutative-probability.html}}}

\Author{Franco FAGNOLA~$^{\rm a}$, Chul Ki KO~$^{\rm b}$ and Hyun Jae YOO~$^{\rm c}$}

\AuthorNameForHeading{F.~Fagnola, C.K.~Ko and H.J.~Yoo}

\Address{$^{\rm a)}$~Mathematics Department, Politecnico di Milano, Piazza L. da Vinci 32, I-20133 Milano, Italy}
\EmailD{\href{mailto:franco.fagnola@polimi.it}{franco.fagnola@polimi.it}}
\URLaddressD{\url{https://www.mate.polimi.it/qp/}}

\Address{$^{\rm b)}$~University College, Yonsei University, 85 Songdogwahak-ro, Yeonsu-gu, Incheon 21983, Korea}
\EmailD{\href{mailto:kochulki@yonsei.ac.kr}{kochulki@yonsei.ac.kr}}

\Address{$^{\rm c)}$~School of Computer Engineering and Applied Mathematics,\\
\hphantom{$^{\rm c)}$}~Institute for Integrated Mathematical Sciences, Hankyong National University,\\
\hphantom{$^{\rm c)}$}~327 Jungang-ro, Anseong-si, Gyeonggi-do 17579, Korea}
\EmailD{\href{mailto:yoohj@hknu.ac.kr}{yoohj@hknu.ac.kr}}

\ArticleDates{Received February 07, 2022, in final form April 19, 2022; Published online May 11, 2022}

\Abstract{We consider an open quantum system with Hamiltonian $H_S$ whose spectrum is given by a generalized Fibonacci sequence weakly coupled to a Boson reservoir in equilibrium at inverse temperature~$\beta$. We find the generator of the reduced system evolution and explicitly compute the stationary state of the system, that turns out to be unique and faithful, in terms of parameters of the model. If the system Hamiltonian is generic we show that convergence towards the invariant state is exponentially fast and compute explicitly the spectral gap for low temperatures, when quantum features of the system are more significant, under an additional assumption on the spectrum of~$H_S$.}


\Classification{81S22; 81S05; 60J80}

\renewcommand{\thefootnote}{\arabic{footnote}}
\setcounter{footnote}{0}

\section{Introduction}

Deformations of CCR and CAR have been extensively investigated in the literature. $q$-deformed commutation relations are
defined by means of a single parameter $q$ in the interval $[-1,1]$ of annihilation and creation operators $a$ and $a^\dagger$
satisfying $a a^\dagger - q a^\dagger a = \unit$ and the CCR or CAR are recovered in the limit as $q\to 1^{-}$ or $q \to -1^{+}$
(see \cite{Bo,BoKuSp,BLW,Mouayn} and the references therein).

Recently, the inclusion of two distinct deformation parameters $r$,~$q$ has been proposed to allow more flexibility while
retaining the good properties and the possibility of finding explicit formulas as in the case of single parameter deformations
(see \cite{Bl,DaKi,GaKa,BM} and the references therein). The two parameters deformed commutation relations
become $a a^\dagger - r a^\dagger a = q^N$, $a a^\dagger - q a^\dagger a = r^N$ where $N$ is the number operator
(see Section \ref{sect:Fibonacci_osc} for precise definitions) in their one-mode Fock space representation.

In this way one finds a quantum system with Hamiltonian $H_S=a^\dagger a$ which is a two parameter deformation of the harmonic
oscillator whose spectrum $\{(r^n-q^n)/(r-q)\}_{n\geq 0}$ is a generalized Fibonacci sequence (that turns out to be the
well-known Fibonacci sequence for $r=\big(1+\sqrt{5}\big)/2$, $q=\big(1-\sqrt{5}\big)/2$) and therefore is called Fibonacci Hamiltonian.
Two parameters Hermite polynomials have been computed and the energy spectrum has been studied showing that the deformation
is more effective in highly excited states (see~\cite{GaKa,BM}). Deformed Fock spaces and deformed Gaussian processes have
been analyzed in connection with the single-parameter deformation of the full Fock space of free probability~\cite{Bl}.
Moreover, the arising quantum algebra with two deformation parameters has been considered in applications to certain physical
models~\cite{DaKi,GaKa}.

In this paper we consider the $r$,~$q$ deformed oscillator with Hamiltonian $H_S$ weakly coupled to a boson reservoir at inverse
temperature $\beta$ as an open quantum system. At first, we rigorously deduce the reduced dynamics of the open system in the weak
coupling limit \cite{AcLuVo} getting a quantum Markov semigroup (QMS) which is a natural two-parameter deformation of the so-called
quantum Ornstein--Uhlenbeck semigroup~\cite{CaMa,CiFaLi} and has a generator in the well-known Gorini--Kossakowski--Lindblad--Sudharshan (GKLS)
form, generalized to allow unbounded operators in the case where one of the parameters is bigger than $1$ that presents more
difficulties (see~\cite{Bo}). We emphasize that, as the reader may immediately note from the proof of
Theorem~\ref{th:unique-inv-st}, the choice of a parameter bigger than $1$ is necessary in order to find an equilibrium
state for the dynamics of the reduced system in order not to break the physical principle of thermal relaxation~\cite{AGGQ,AcLuVo}.

In our analysis, we pay a special attention to the structure of the spectrum of $H_S$
whose order plays an important role in determining the GKLS generator motivating the emergence of natural inequalities
among the two parameters~$r$,~$q$. In particular, the key conditions $-1\leq q\leq 1<r$ and $r+q\geq 1$ that
appear throughout the paper, are not for mere convenience because they affect the order of the spectrum and, as a consequence,
the QMS that emerges after the weak coupling limit. However, they still allow us to analyze the behaviour of the system for parameters
$r\to 1^+$ and $q\to 1^{-}$, when the $r$, $q$ deformed commutation are ``near'' the CCR, for comparison with the quantum
Ornstein--Uhlenbeck semigroup~\cite{CaMa,CiFaLi}.

We then focus on the case where the spectrum of $H_S$ is generic (see last part of Section~\ref{sect:QMS-WCLT} for the precise
definition) in which the GKLS generator takes a simpler form (see \cite{CaSaUm,FV} and the references therein).
We emphasize that this happens for almost all choices of the deformation parameters $r$, $q$ (in the sense of Lebesgue measure on $\mathbb{R}^2$, for instance). We show that the arising QMS has a normal invariant state, which is unique and faithful, and investigate
the speed of convergence towards the invariant state determining the spectral gap in the $L_2$ space of the invariant state.

Taking advantage of the structure of the GKLS generator we are able to compute explicitly the spectral gap for low temperatures,
when quantum features of the dynamics are more significant, if $r+q\geq 2$ (Theorem~\ref{thm:sp-gap-all}).
We also provide evidence (see Remark~\ref{rem:j=0&k=1!}) that the spectral gap is strictly positive but it is not possible
to obtain a simple closed-form expression for high temperatures.

The case where a parameter $r$, $q$ is strictly bigger than $1$ is the most difficult when considering deformations of
the CCR (see~\cite{Bo}) not only because of unboundedness of creation and annihilation operators, as in the boson case, but also
because of additional pathologies that arise. It is well-known, for instance, that field operators are not essentially self-adjoint
on the domain of finite particle vectors. However, our results complement those obtained for special cases $q=r=1$ (Boson),
$r=1$, $q=-1$ (Fermi) and $r=1$, $q=0$ scattered in the literature.

In addition the computation of the spectral gap of a GKLS generator has its own interest because of applications
in the study of strong ergodicity of open quantum systems \cite{BeDaSa,DFY,MuAl20,MuAl22} and the explicit result is
known only in few cases.

The paper is organized as follows. In Section \ref{sect:Fibonacci_osc} we discuss the structure of the spectrum of
generalized Fibonacci oscillators in order to justify the emergence of conditions on parameters~$r$,~$q$ that will be
assumed in the paper. The deduction from the weak coupling limit of form generators of QMSs for generalized Fibonacci Hamiltonians
viewed as open quantum systems is illustrated in Section \ref{sect:QMS-WCLT} and the construction of QMSs from form
generators by the minimal semigroup method (see \cite[Section~3]{FF-Proy}) is presented in Section \ref{sect:gen-open-Fib}.
The spectral gap is computed in Section \ref{sect:sp-gap} in a simple explicit formula for small temperatures of the reservoir
and for parameters satisfying $-1<q \leq 1 < r$, $r+q\geq 2$ also providing evidence that an explicit formula in the
general case cannot be achieved.

\section{Fibonacci oscillators}\label{sect:Fibonacci_osc}

Let $q$, $r$ be two real numbers with $q\not=r$. $(q,r)$-integers are defined by
\begin{equation*}
\varepsilon_0=0,\qquad \varepsilon_1=1,\qquad \varepsilon_n=\frac{r^n - q^n}{r-q}\quad \text{for\ } n\ge 2,
\end{equation*}
where we can assume $r>q$ without loss of generality. In the case where $q\to 1^{-}$ and $r\to 1^{+}$ one
finds the natural numbers and, if $r=1$ the usual $q$-integers (also allowing $q>1$). It is worth noticing
that, in the special case where $r=\big(\sqrt{5}+1\big)/2$, $q=\big(1-\sqrt{5}\big)/2$, one finds the sequence of Fibonacci
numbers. For this reason $(\varepsilon_n)_{n\geq 0}$ is called generalized Fibonacci sequence.
Note that $\varepsilon_n\geq 0$ for all $n\geq 0$.

Let $\mathsf{h}=\Gamma(\mathbb{C})$ be the one-mode Fock space with canonical orthonormal basis $(e_n)_{n\ge 0}$.
The Fibonacci oscillator is the quantum system with Hamiltonian
\begin{equation}\label{eq:Fibonacci-H}
H_S = \sum_{n\geq 0} \varepsilon_n |e_n\rangle\langle e_n|,
\end{equation}
whose spectrum is the generalized Fibonacci sequence. Defining the $(q,r)$ annihilation and creation operators
\begin{gather}
\operatorname{Dom}(a) =
\bigg\{ u\in\mathsf{h} \,\bigg|\, \sum_{n\ge 0} \varepsilon_n |u_n|^2 <\infty \bigg\},
\qquad a e_n =\sqrt{\varepsilon_n} e_{n-1},
\label{eq:annihil}\\
\operatorname{Dom}\big(a^\dagger\big) =
\bigg\{ u\in\mathsf{h} \,\bigg|\, \sum_{n\ge 0} \varepsilon_{n+1} |u_n|^2 <\infty \bigg\},
\qquad
a^\dagger e_n = \sqrt{\varepsilon_{n+1}} e_{n+1},\nonumber 
\end{gather}
one can write $H_S=a^\dagger a$ on the domain $F$ of finite linear combinations of vectors of the
canonical orthonormal basis, also called finite particle vectors.

One immediately checks that $a$ and $a^\dagger$ are bounded operators if and only if $-1\leq q < r\leq 1$, they
are mutually adjoint and satisfy the commutation relations
\begin{equation}\label{eq:rq-comm}
a a^\dagger - r a^\dagger a = q^N, \qquad a a^\dagger - q a^\dagger a = r^N,
\end{equation}
where $N$ is the usual number operator defined by
\[
\operatorname{Dom}(N)=\bigg\{ u\in\mathsf{h} \,\bigg|\, \sum_{n\ge 0} n^2 |u_n|^2 <\infty \bigg\}, \qquad
N u = \sum_{n\ge 1} nu_n e_n.
\]
Paying attention to the operator domains these properties can be extended to the general case $q<r$.
In particular, for $r=1$, we have the $q$-commutation relations $a a^\dagger - q a^\dagger a= \unit$.
These commutation relations can be found also considering creations and annihilations on interacting
Fock spaces (see~\cite{AcLu,GeSk} and the references therein) but it is more convenient to consider
the usual one-mode Fock space for our analysis. Moreover, we would like to mention
that two parameter deformed commutation relations lead to remarkable combinatorial formulas (e.g., for
moments of field operators) as those of canonical commutation and anti-commutation relations (see~\cite{GoKu}).

Since we are interested in the generalized Fibonacci oscillator as an open quantum
system weakly coupled with a reservoir and the weak coupling crucially depends on ordering of the
spectrum of~$H_S$, throughout the paper assume that eigenvalues $\varepsilon_n$ of $H_S$ form
an increasing sequence. Clearly, this is not the case, for example, if $0<q<r<1$ because
$r^{n+1}-q^{n+1}< r^n-q^n$ for big $n$. Moreover, in order to exclude high oscillatory behaviours
also for reasons that will be clear in the next section, we are mostly interested in the case
$-1\leq q \leq 1 < r$ therefore this inequality will also be assumed throughout the paper.

Note that $\varepsilon_2\geq \varepsilon_1$ if and only if $r^2-q^2\geq r-q$, i.e., $r+q-1\geq 0$,
therefore we need at least this additional condition. Once it holds, the sequence $(\varepsilon_n)_{n\geq 0}$ is
obviously increasing if $r>1>q\geq 0$ because
\[
r^{n+1}-q^{n+1}-(r^n-q^n)=r^n(r-1)+q^n(1-q)>0.
\]
The case $-1\leq q < 0 < 1 < r$ needs a slightly more detailed analysis. First of all note that, since $r$,~$q$
solve the equation $0=(x-r)(x-q)=x^2 - (r+q) x +rq$ and $\varepsilon_0=0,\varepsilon_1=1$ then
$\varepsilon_{n+2} -(r+q)\varepsilon_{n+1} + rq \varepsilon_{n}=0$.
 Therefore, since $r+q-1\geq 0 $ and $-rq>0$, the identity
\[
\varepsilon_{n+2}-\varepsilon_{n+1} = (r+q-1)\varepsilon_{n+1}-rq\,\varepsilon_n
\]
shows that the sequence $(\varepsilon_n)_{n\geq 0}$ is non decreasing whenever $r+q-1\geq 0$.

The above discussion is summarized by the following
\begin{Lemma}\label{lem:en-nondecr}
Assume $-1\leq q \leq 1 < r$.
The sequence $(\varepsilon_n)_{n\geq 0}$ is non-decreasing $($resp.\ strictly increasing$)$
if and only if $r+q-1\geq 0$ $($resp.\ $r+q-1 > 0)$.
\end{Lemma}
\begin{Remark} The ``free'' $r=1$, $q=0$ and Fibonacci $r=\big(\sqrt{5}+1\big)/2$, $q=\big(1-\sqrt{5}\big)/2$ cases lie
on the boundary of the region. Note that the $q=r=1$ (Bose) and $q=1=-r$ (Fermi) cases, are formally excluded,
but can arise as limiting cases. However, the spectrum of system Hamiltonian is no more generic and
the study of QMS arising from the weak coupling limit has been carried on separately~\cite{CiFaLi}.
\end{Remark}

\section{QMS of weak coupling limit type}\label{sect:QMS-WCLT}

Let $H_S$ be a Hamiltonian with spectral decomposition
\begin{equation*}
H_S=\sum_{m\geq 0}
\varepsilon_{m} P_{\varepsilon_m},
\end{equation*}
where $\varepsilon_m$, with $\varepsilon_m<\varepsilon_n$ for $m<n$, are the eigenvalues
of $H_S$ and $P_{\varepsilon_m}$ are the corresponding eigenspaces.
QMSs of \textit{weak coupling limit type} (WCLT), associated with the Hamiltonian~$H_S$
(see \cite{AcLuVo,FQ} and the references therein), have generators of the form
${\mathcal L} = \sum_{\omega\in B_{+}} {\mathcal L}_{\omega}$ where $B_{+}$ is the set of all Bohr frequencies (Arveson spectrum)
\begin{equation*}
B_{+} :=\{\varepsilon_{n} - \varepsilon_{m} \colon \varepsilon_{n} - \varepsilon_{m} >0 \}.
\end{equation*}
Given a system operator $D$, whose domain contains ranges of projections $P_{\varepsilon_m}$,
depending on the interaction of the system with a reservoir, for every Bohr frequency $\omega$,
consider a generator ${\mathcal L}_{\omega}$ with the
Gorini--Kossakowski--Lindblad--Sudarshan (GKLS) structure
\begin{gather}
{\mathcal L}_{\omega}(x) = \mathrm{i} [H_\omega, x] -\frac{\Gamma^{-}_{\omega}}{2}
\left( D_{\omega}^{*}D_{\omega}x -2D_{\omega}^{*}xD_{\omega} + xD_{\omega}^*D_{\omega}\right)
\nonumber \\
 \hphantom{{\mathcal L}_{\omega}(x) =}{}
 - \frac{\Gamma^{+}_{\omega}}{2}
\left( D_{\omega}D_{\omega}^*x -2D_{\omega}xD_{\omega}^{*} + xD_{\omega}D_{\omega}^{*}\right)\label{eq:GKLS-WCLT}
\end{gather}
for all $x\in{\mathcal B}(\mathsf{h})$, with Kraus operators $D_{\omega}$ defined by
\begin{equation*}
D_{\omega}=
\sum_{(\varepsilon_{n}, \varepsilon_{m})\in B_{+,\omega}}P_{\varepsilon_{m}} D P_{\varepsilon_{n}},
\end{equation*}
where $B_{+,\omega}= \{(\varepsilon_{n} , \varepsilon_{m}) \,|\,
\varepsilon_{n} - \varepsilon_{m} =\omega \}$,
$\Gamma^{-}_{\omega}$, $\Gamma^{+}_{\omega}$ are non-negative
real constants with $\Gamma^{-}_{\omega}+\Gamma^{+}_{\omega}>0$
and $H_\omega$ is a bounded self-adjoint operator on~$\mathsf{h} $ commuting with~$H_S$.

In the case when the set of Bohr frequencies is infinite, for $\mathcal{L}$ to be the
generator of a norm continuous QMS the series
\begin{equation*}
\sum_{\omega\in B_{+}}\left( \Gamma^{-}_{\omega}D_\omega^*D_\omega
+\Gamma^{+}_{\omega} D_\omega D_\omega^*\right)
\end{equation*}
must be strongly convergent in ${\mathcal B}(\mathsf{h})$, the von Neumann algebra of all bounded operators
on~$\mathsf{h}$ (see \cite[Corollary~30.13, p.~268 and Theorem~30.16, p.~271]{KRP}).

QMSs generators in this form arise in the weak coupling limit of a system with Hamiltonian~$H_S$ coupled to a
Boson reservoir in equilibrium at inverse temperature $\beta$ with coupling
\[
D\otimes A^\dagger(\phi) + D^*\otimes A(\phi),
\]
where $A^\dagger(\phi), A(\phi)$ are the creation and annihilation operator of the reservoir with
test function~$\phi$. Constants $\Gamma^\pm_{\omega}=f_{\omega}\gamma^\pm_{\omega}$ are given explicitly by
\begin{equation}\label{eq:Gamma-omega}
\gamma^{-}_{\omega} = \frac{{\rm e}^{\beta\omega}}{{\rm e}^{\beta\omega}-1}, \qquad
\gamma^{+}_{\omega} = \frac{1}{{\rm e}^{\beta\omega}-1}, \qquad
f_{\omega} = \int_{\{y\in\mathbb{R}^3\,|\, |y|=\omega \}} |\phi(y)|^2 {\rm d}_s y,
\end{equation}
where ${\rm d}_s$ denotes the surface integral and the cut-off $\phi$ is a square-integrable function on~$\mathbb{R}^3$.
A~realistic cut-off could be a function which is constant in some bounded ``big'' region and slowly vanishes as~$\omega$ goes to infinity. For this reason, slightly modifying generators $\mathcal{L}_\omega$ after the weak coupling limit,
if necessary, it looks reasonable to assume throughout the paper~$f_\omega$ constant and fix~$f_\omega=1$.

From the above discussion it is clear that the spectral decomposition of the Hamiltonian $H_S$ plays a key role.
In particular, if we consider the case $r=1$, $q=0$ (in which (\ref{eq:rq-comm}) are the well-known free commutation relations)
the Hamiltonian $H_S$ becomes
\[
H_S = a^\dagger a = \sum_{n\geq 1} |e_n\rangle\langle e_n| = P_1
\]
with $P_1$ projection. Therefore there is only the Bohr frequency $\omega=1$,
\[
D_1 = |e_0\rangle\langle e_0| D P_1 = \left|e_0\right\rangle\left\langle P_1 D^*e_0\right|
\]
is determined by the vector $v=P_1 D^*e_0$ orthogonal to $e_0$ so that
\[
D_1=\left|e_0\right\rangle\left\langle v\right|, \qquad D_1^*=\left|v\right\rangle\left\langle e_0\right|,
\qquad D_1^*D_1 = \left|v\right\rangle\left\langle v\right|, \qquad
D_1 D_1^*=\Vert v\Vert^2 \left|e_0\right\rangle\left\langle e_0\right|
\]
and $H_1$ is a multiple of $P_1$ up to addition of a multiple of the identity operator.
In this way, calling $e$ the normalized vector $v$, we get the GKLS generator
\begin{gather*}
{\mathcal L}_{\omega}(x)
 = \mathrm{i}\kappa [\,|e\rangle\langle e|, x] -\frac{\Gamma^{-}_{1}\Vert v\Vert^2}{2}
\left( \left|e\right\rangle\left\langle e\right| x
-2 \left|e\right\rangle\left\langle e_0\right| x \left|e_0\right\rangle\left\langle e\right|
+ x \left|e\right\rangle\left\langle e\right|\right)
\nonumber \\
 \hphantom{{\mathcal L}_{\omega}(x)=}{}
 - \frac{\Gamma^{+}_{1}\Vert v\Vert^2}{2}
\left( \left|e_0\right\rangle\left\langle e_0\right|x
-2 \left|e_0\right\rangle\left\langle e\right| x \left|e \right\rangle\left\langle e_0\right|
+ x \left|e_0\right\rangle\left\langle e_0\right|\right).
\end{gather*}
This GKLS generator essentially describes the dynamics of a $2$-level system than can be
explicitly computed. One finds the same dramatic simplification for $r=1$, $q=-1$ and so we are not interested
in these special cases.

It is worth noticing that the GKLS generator written down just by analogy with the Boson case,
namely
\[
{\mathcal L}(x)
= \mathrm{i}\kappa [\,a^\dagger a, x] -\frac{\Gamma^{-}_1}{2}
\big( a a^\dagger x - 2 a x a^\dagger + x a a^\dagger x\big)
- \frac{\Gamma^{+}_1}{2} \big(a^\dagger a x -2 a^\dagger x a + x a^\dagger a\big)
\]
has another structure.

It is well-known that, when the system Hamiltonian $H_S$ is \emph{generic} namely:
\begin{itemize}\itemsep=0pt
\item[(i)] Its spectrum is pure point and each eigenvector has multiplicity one,
\item[(ii)] For all $\omega\in B_{+}$ there exists a unique pair $(n,m)$ of energy levels
such that $\varepsilon_n-\varepsilon_m=\omega$,
\end{itemize}
the structure of the generator is very simple (see, e.g.,~\cite{FV}) because operators
$D_\omega$ are multiples of rank one operators $|e_m\rangle\langle e_n|$, where $n$, $m$
are determined by the unique pair of such that $\varepsilon_n-\varepsilon_m=\omega$. In particular each off-diagonal
rank one operator $|e_i\rangle\langle e_j|$ is an eigenvector for~$\mathcal{L}$ and the
action of maps $\mathcal{T}_t$ of the QMS $(\mathcal{T}_t)_{t\geq 0}$ generated by~$\mathcal{L}$
(see Section~\ref{sect:gen-open-Fib} for details) on $|e_i\rangle\langle e_j|$ is explicit.

The Fibonacci type Hamiltonian $H_S$ as in~(\ref{eq:Fibonacci-H}) is clearly generic for almost
all choices of parameters~$r$,~$q$.
However, in other cases, the WCLT generator might be more complex because of the structure of~$B_{+}$
(see~\cite{DFSU} for a detailed analysis of the structure of norm-continuous QMSs).
Indeed, if $(\varepsilon_n)_{n\geq 0}$ is the Fibonacci sequence, then $\varepsilon_0=0$, $\varepsilon_1=1$
and $\varepsilon_{n+1}=\varepsilon_n+\varepsilon_{n-1}$ for $n\geq 1$ so that
$B_{+}= \{ \varepsilon_n\,|\, n\geq 1 \}$ because, for all $n\geq 1$
\[
\varepsilon_{n+1}-\varepsilon_{n}= \varepsilon_{n}-\varepsilon_{n-2} = \varepsilon_{n-1}
\]
and, for $k\geq 3$,
\[
\varepsilon_{n+k}-\varepsilon_{n} = \varepsilon_{n+k-1}+\varepsilon_{n+k-2}-\varepsilon_{n}
 = 2\varepsilon_{n+k-2}+\varepsilon_{n+k-3}-\varepsilon_{n}
 \geq 2\varepsilon_{n+k-2} > \varepsilon_{n'}-\varepsilon_{m'}
\]
for all $n'<n+k-2$ and so, in particular, for all $n'< n$. In this case, as a consequence,
\begin{gather*}
D_{\varepsilon_1}=D_1 = \langle e_0,D e_1\rangle |e_0\rangle\langle e_1|
+ \langle e_0,D e_2\rangle |e_0\rangle\langle e_2| \\
\hphantom{D_{\varepsilon_1}=D_1 =}{} + \langle e_1,D e_3\rangle |e_1\rangle\langle e_3|
+\langle e_2,D e_3\rangle |e_2\rangle\langle e_3|
+ \langle e_3,D e_4\rangle |e_3\rangle\langle e_4|.
\end{gather*}
Clearly, the operator $D$ also plays a key role in the GKLS generator~$\mathcal{L}$ because transitions between levels
$\varepsilon_n$ and $\varepsilon_m$ $(\varepsilon_n-\varepsilon_m=\omega>0$) can be forbidden if $\langle e_m, D e_n\rangle$
is zero even if $\Gamma^{\pm}_\omega$ is strictly positive.
The most natural choice for $D$ is the annihilator $D=a$ defined by~(\ref{eq:annihil}). With this choice
of $D$ and the Fibonacci sequence as $(\varepsilon_n)_{n\geq 0}$ we find
\[
D_{\varepsilon_1} = |e_0\rangle\langle e_1| + \sqrt{2} |e_2\rangle\langle e_3| + \sqrt{3} |e_3\rangle\langle e_4|.
\]
However, with other choices of the operator $D$, Kraus operators $D_\omega$ can be rank one also
in the case where $(\varepsilon_n)_{n\geq 0}$ is the Fibonacci sequence.

\section{Generic open Fibonacci type oscillators}\label{sect:gen-open-Fib}

From now on we consider the GKLS form generator
\begin{equation}\label{eq:L-Bohr}
\mathcal{L}=\sum_{n\geq 1}\mathcal{L}_n,
\end{equation}
where
\begin{gather*}
 \mathcal{L}_n(x) = \mathrm{i} \left[ \kappa^{-}_n |e_n\rangle\langle e_n|
 +\kappa^{+}_n |e_{n-1}\rangle\langle e_{n-1}|,x \right] \\
\hphantom{\mathcal{L}_n(x) =}{} -
 \frac{\Gamma^{-}_n}{2}\frac{r^n-q^n}{r-q}
 \left( |e_n\rangle\langle e_n| x - 2|e_n\rangle\langle e_{n-1}| x |e_{n-1}\rangle\langle e_n|
 + x|e_{n}\rangle\langle e_n|\right) \nonumber \\
\hphantom{\mathcal{L}_n(x) =}{}-
 \frac{\Gamma^{+}_n}{2}\frac{r^n-q^n}{r-q}
 \left( |e_{n-1}\rangle\langle e_{n-1}| x
 - 2|e_{n-1}\rangle\langle e_n| x |e_n\rangle\langle e_{n-1}|
 + x|e_{n-1}\rangle\langle e_{n-1}|\right),
\end{gather*}
$\Gamma^{\pm}_n$ are as in (\ref{eq:Gamma-omega}) and $\kappa^{\pm}_n$ are real constants.
As explained in the previous section, these are generators~(\ref{eq:GKLS-WCLT}) where
all operators $D_\omega$ are rank-one because either $H_S$ is generic or by suitable choice
of the operator~$D$. The set of $\omega$ with terms contributing to the generator is in one-to-one correspondence
with $\mathbb{N}^*=\{1,2,3,\dots\}$ and transitions from level $\varepsilon_n$ ($n>0$) can occur only to levels
$\varepsilon_{n+1}$ and $\varepsilon_{n-1}$ for $n>0$ and from level $\varepsilon_0=0$ to $\varepsilon_1=1$ so that
the graph of the process is as follows:
\begin{figure}[h!]\centering
$$
\xymatrix{*++[o][F]{\varepsilon_0=0} \ar@(ur,ul)[rr]
&& *++[o][F]{\varepsilon_1=1} \ar@(ur,ul)[rr] \ar@(dl,dr)[ll]
&& *++[o][F]{\varepsilon_2} \ar@(dl,dr)[ll] \ar@(ur,ul)[rr]
&& *++[o][F]{\varepsilon_3} \ar@(dl,dr)[ll] \ar@(ur,ul)[rr]
&& *++[ ][F]{\dots} \ar@(dl,dr)[ll] }
$$
{\small Graph of the nearest neighbour jump process.}
\end{figure}

The definition of $\mathcal{L}$ is only formal because the sum on $n$ in~(\ref{eq:L-Bohr})
is infinite, therefore some clarifications are now in order. First of all note that the operator
\[
G = \sum_{n\geq 0}\left(\mathrm{i}(\kappa^{-}_n+\kappa^{+}_{n+1})
-\frac{\Gamma^{-}_n}{2}\frac{r^n-q^n}{r-q}-\frac{\Gamma^{+}_{n+1}}{2}\frac{r^{n+1}-q^{n+1}}{r-q}\right)|e_n\rangle\langle e_n|
\]
is well defined as a normal operator on the domain Dom$(G)$ of vectors $u=\sum_n u_n e_n\in\mathsf{h}$,
i.e., such that $\sum_{n\geq 0}|u_n|^2 < \infty$ for which
\[
\sum_{n\geq 0} \left(\left| \kappa^{-}_n+\kappa^{+}_{n+1}\right|^2
+\left|\frac{\Gamma^{-}_n}{2}\frac{r^n-q^n}{r-q}
+\frac{\Gamma^{+}_{n+1}}{2}\frac{r^{n+1}-q^{n+1}}{r-q}\right|^2\right) |u_n|^2 < \infty.
\]
In particular, if sequences $(\kappa^{-}_n+\kappa^{+}_n)_{n\geq 1}$ and $((\Gamma^{-}_n+\Gamma^{+}_n)(r^n-q^n))_{n\ge 1}$
are bounded, then $G$ is bounded, $\mathcal{L}$ is a bounded operator on $\mathcal{B}(\mathsf{h})$ and generates a norm
continuous QMS on $\mathcal{B}(\mathsf{h})$ with Kraus operators $L_\ell$ $(\ell\geq 1)$ which are rank-one and given by
\begin{equation*}
L_{2\ell} = \left(\frac{\Gamma^{-}_\ell(r^\ell-q^\ell)}{r-q}\right)^{1/2} \kern-8truept |e_{\ell-1}\rangle\langle e_\ell|, \qquad
L_{2\ell+1} = \left(\frac{\Gamma^{+}_\ell(r^\ell-q^\ell)}{r-q}\right)^{1/2} \kern-8truept |e_\ell\rangle\langle e_{\ell-1}|.
\end{equation*}

However, even if $G$ is unbounded as in typical cases with $r>1$, it is possible to construct a~uniquely determined QMS
on $\mathcal{B}(\mathsf{h})$ by the minimal semigroup method (see, e.g., \cite[Sections~3.3 and~3.4]{FF-Proy} and also~\cite{KuSiSr} and the references therein). Indeed, $G$~generates a strongly continuous semigroup $(P_t)_{t\geq 0}$ on $\mathsf{h}$ and the explicit form of the operators~$P_t$
is immediately written. For $x\in\mathcal{B}(\mathsf{h})$ let $\pounds(x)$ be the quadratic form with domain
$\operatorname{Dom}(G)\times\operatorname{Dom}(G)$
\begin{equation}\label{eq:Lform}
\pounds(x)[v,u] = \left\langle Gv, x u \right\rangle
+\sum_{\ell\geq 1}\left\langle L_\ell v, x L_\ell u \right\rangle
+ \left\langle v, x G u \right\rangle.
\end{equation}

The {\em minimal} semigroup associated with operators $G$, $L_\ell$ is constructed, on elements $x$ of ${\mathcal{B}}(\mathsf{h})$,
by means of the non decreasing sequence of positive maps $({\mathcal{T}}^{(n)}_t)_{n\ge 0}$ defined, by recurrence, as follows
\begin{gather}
{\mathcal{T}}^{(0)}_t(x) = P_t^* xP_t, \label{eq:min-semigr-induct}\\
\big\langle v, {\mathcal{T}}^{(n+1)}_t(x)u \big\rangle
 = \langle P_t v, x P_tu \rangle
+\sum_{\ell\geq 1}\int_0^t
\big\langle L_\ell P_{t-s}v,
{\mathcal{T}}^{(n)}_s(x)L_\ell P_{t-s}u \big\rangle \, {\rm d}s \nonumber
\end{gather}
for all $x\in{\mathcal{B}}(\mathsf{h})$, $t\ge 0$, $v,u\in\operatorname{Dom}(G)$.
Indeed, we have
\[
{\mathcal{T}}_t(x)=\sup_{n\ge 0}{\mathcal{T}}^{(n)}_t(x)
\]
for all positive $x\in{\mathcal{B}}(\mathsf{h})$ and all $t\ge 0$. The definition of
positive maps ${\mathcal{T}}_t$ is then extended to all the elements of
${\mathcal{B}}(\mathsf{h})$ by linearity.
The minimal semigroup associated with $G$, $L_\ell$ satisfies the integral equation
\begin{equation} \label{eq-semigr-int-equat}
\langle v,{\mathcal{T}}_t(x) u\rangle = \langle v, x u\rangle + \int_0^t \pounds(\mathcal{T}_s(x))[v,u]\, {\rm d}s
\end{equation}
for all $x\in{\mathcal{B}}(\mathsf{h})$, $t\ge 0$, $v,u\in\operatorname{Dom}(G)$.
Moreover, it is the unique solution to the above equation if and only if it is conservative (or Markov), i.e.,
${\mathcal{T}}_t(\unit)=\unit$ for all $t\ge 0$ (see, e.g., \cite[Corollary~3.23]{FF-Proy}).

In our framework it is not difficult to show that conservativity is equivalent to the Karlin--McGregor condition
for non-explosion of Markov jump processes. Indeed, one immediately checks that the diagonal algebra generated by
projections $|e_n\rangle\langle e_n|$ is invariant for maps $\mathcal{T}_t^{(n)}$ (for all $n\geq 0$) defined recursively by
(\ref{eq:min-semigr-induct}) because each vector $e_n$ is an eigenvector of $G$ so that
$P^*_t |e_j\rangle\langle e_j|P_t = {\rm e}^{z_j t}|e_j\rangle\langle e_j|$ for some $z_j\in\mathbb{C}$ and, looking
at iterations (\ref{eq:min-semigr-induct}), if $\mathcal{T}^{(n)}_s(|e_j\rangle\langle e_j|)$ belongs to the diagonal
algebra, then
\[
P_{t-s}^*L_\ell^*\mathcal{T}^{(n)}_s(|e_j\rangle\langle e_j|)L_\ell P_{t-s}
\]
belongs to the diagonal algebra as well for all $0\leq s\leq t$, and so also $\mathcal{T}^{(n+1)}_s(|e_j\rangle\langle e_j|)$
belongs to the diagonal algebra. It follows that $\mathcal{T}_s(|e_j\rangle\langle e_j|)$ belongs to the diagonal
algebra which is invariant for the QMS~$\mathcal{T}$, as expected from the quadratic form computation
\[
\pounds(f(a^\dagger a))
= \sum_{n\geq 1} \Gamma^{-}_n \epsilon_n \left( f(\epsilon_{n-1})-f(\epsilon_n) \right) |e_n\rangle\langle e_n|
+ \sum_{n\geq 0} \Gamma^{+}_{n+1} \epsilon_{n+1} \left( f(\epsilon_{n+1})-f(\epsilon_n) \right) |e_n\rangle\langle e_n|
\]
for all bounded function $f$ on the spectrum of $H_S$.
In this way, we see that the restriction of maps~$\mathcal{T}_t$ to the diagonal algebra coincides with
the minimal semigroup of the classical birth-and-death process with birth (resp.\ death) rates $\lambda_n$ (resp. $\mu_n$)
\begin{equation}\label{eq:lambda-mu}
\lambda_n = \Gamma^{+}_{n+1} \frac{r^{n+1}-q^{n+1}}{r-q} =\Gamma^{+}_{n+1}\varepsilon_{n+1},
\qquad \mu_n = \Gamma^{-}_n \frac{r^n-q^n}{r-q}=\Gamma^{-}_n\varepsilon_n.
\end{equation}
Moreover, writing (\ref{eq-semigr-int-equat}) for $x=|e_j\rangle\langle e_j|$, $u=v=e_i$ and recalling that
$\langle e_i, \mathcal{T}_t(|e_j\rangle\langle e_j|) e_i\rangle$ is the probability of visiting $j$ at time $t$
starting from~$i$ at time~$0$, we find the backward Kolmogorov equations of the birth-and-death process.
Therefore the minimal semigroup is Markov if and only if the minimal semigroup of
the classical birth-and-death process with the above rates is conservative.

Let $\pi_0=1$ and, for $n\geq 1$,
\begin{equation*}
\pi_n = \left(\lambda_0\lambda_1\cdots\lambda_{n-1}\right)/\left(\mu_1\mu_2\cdots\mu_n\right) ={\rm e}^{-\beta\varepsilon_n}.
\end{equation*}
It is well-known \cite[Theorem 2.2, p.~100]{An} that the minimal semigroup of the classical birth-and-death process
with these rates is conservative (more precisely, the minimal semigroup is identity preserving and it is the unique solution
of the backward Kolmogorov equations) if and only if
\begin{equation*} 
\sum_{n\geq 1}\frac{1}{\lambda_n\pi_n}\sum_{k=1}^{n} \pi_k = + \infty.
\end{equation*}

\begin{Theorem}\label{th:minQMSunital}The minimal semigroup associated with the above $G$, $L_\ell$ is Markov for all $-1\leq q\leq 1< r$
with $r+q\geq 1$.
\end{Theorem}

\begin{proof} Notice that $\omega_n=\epsilon_n-\epsilon_{n-1}$ and so,
by~(\ref{eq:Gamma-omega}), we have
\[
\Gamma^{+}_n = \frac{1}{{\rm e}^{\beta\omega_n}-1}, \qquad
\Gamma^{-}_n = \frac{{\rm e}^{\beta\omega_n}}{{\rm e}^{\beta\omega_n}-1}.
\]
Thus we get $\pi_n ={\rm e}^{-\beta\varepsilon_n}$ for all $n\geq 0$ and
\[
\frac{1}{\lambda_n\pi_n} = \frac{{\rm e}^{\beta\varepsilon_{n+1}}-{\rm e}^{\beta\varepsilon_n}}{\varepsilon_{n+1}}
=\frac{{\rm e}^{\beta\varepsilon_n}}{\varepsilon_{n+1}}\big({\rm e}^{\beta(r^n(r-1)+q^n(1-q))/(r-q)}-1\big).
\]
Since the sequence $(1/\lambda_n\pi_n)_{n\geq 0}$ diverges as $n$ goes to $+\infty$ and
\[
\sum_{n=1}^\infty \frac{1}{\lambda_n\pi_n}\sum_{k=1}^n\pi_k
> \pi_1 \sum_{n=1}^\infty \frac{1}{\lambda_n\pi_n} = + \infty.
\]
This completes the proof.
\end{proof}

In the sequel we will assume that parameters $-1\leq q\leq 1< r$ satisfy $r+q\geq 1$.
Theorem~\ref{th:minQMSunital} also implies that \cite[Proposition~3.33]{FF-Proy} the domain of the generator
of~$\mathcal{T}$ is the space of $x\in\mathcal{B}(\mathsf{h})$ for which the quadratic form~$\pounds(x)$~(\ref{eq:Lform}) with domain $\operatorname{Dom}(G)\times\operatorname{Dom}(G)$ is bounded. This happens, in particular, for
all off diagonal rank-one operators $|e_j\rangle\langle e_k|$ $(j\not=k$)
\begin{gather}
\mathcal{L}(|e_j\rangle\langle e_k|) = \Bigg( \mathrm{i}(\kappa^{-}_j-\kappa^{-}_{k} +\kappa^{+}_{j+1}-\kappa^{+}_{k+1}) \nonumber\\
\hphantom{\mathcal{L}(|e_j\rangle\langle e_k|) =}{}
- \frac{\Gamma^{-}_j}{2}\varepsilon_j-\frac{\Gamma^{+}_{j+1}}{2}\varepsilon_{j+1}
-\frac{\Gamma^{-}_k}{2}\varepsilon_k-\frac{\Gamma^{+}_{k+1}}{2}\varepsilon_{k+1}\Bigg)|e_j\rangle\langle e_k|,\label{eq:L-off-diag}
\end{gather}
that are eigenvectors of $\mathcal{L}$ with nonzero eigenvalue. This remark allows us to prove in a simple way
existence and uniqueness of a normal invariant state.

\begin{Theorem}\label{th:unique-inv-st}
Suppose that $-1\leq q\leq 1< r$ and $r+q\geq 1$. The QMS admits a unique invariant state $\rho$
\begin{equation}\label{eq:inv-st}
\rho = \frac{1}{Z_\beta}\sum_{n\geq 0} {\rm e}^{-\beta\varepsilon_n} |e_n\rangle\langle e_n|, \qquad
Z_\beta = \sum_{n\geq 0} {\rm e}^{-\beta\varepsilon_n}.
\end{equation}
\end{Theorem}

\begin{proof}
First of all note that $\sum_{n>0}{\rm e}^{-\beta\varepsilon_n}<+\infty$.
Let $\rho$ be a normal invariant state. Since rank one operators $|e_j\rangle\langle e_k|$ $(j\not=k$)
belong to the domain of the generator $\mathcal{L}$ and are eigenvectors with nonzero eigenvalue $\xi_{jk}$, say, differentiating
the identity $\operatorname{tr}(\rho\mathcal{T}_t(|e_j\rangle\langle e_k|))=\operatorname{tr}(\rho |e_j\rangle\langle e_k|)$
at $t=0$, we get
\[
0=\operatorname{tr}\left(\rho\mathcal{L}(|e_j\rangle\langle e_k|)\right)= \xi_{jk}\langle e_k, \rho e_j\rangle
\]
and so $\rho$ is diagonal in the same basis as $H_S$, i.e., $\rho=\sum_{n\geq 0}\rho_n |e_n\rangle\langle e_n|$.
Now a simple computation shows that the probability density $(\rho_n)_{n\geq 0}$ on $\mathbb{N}$ is an invariant
measure for the associated classical birth-and-death process. Therefore (see, e.g., \cite[Example 4.2, p.~197]{An})
the state~$\rho$ given by~(\ref{eq:inv-st}) is invariant because $\rho_n:={\rm e}^{-\beta\varepsilon_n} /Z_\beta$
defines an invariant density for classical birth-and-death process.

Uniqueness follows immediately because we proved that a normal invariant state is diagonal and it determines
an invariant density for the associated classical birth-and-death process which is unique because the
birth-and-death process has strictly positive transition rates whence it is irreducible.
\end{proof}

\begin{Remark}Note that the invariant state (\ref{eq:inv-st}) is faithful.
\end{Remark}

The rest of the paper is dedicated to the study of the speed of convergence of $\mathcal{T}$ towards the invariant state.

\section{Spectral gap}\label{sect:sp-gap}
Strong ergodic properties, such as the speed of convergence towards the invariant state,
are a~natural problem on the behaviour of an open quantum system with a unique faithful normal invariant state.
In this section we discuss the spectral gap of the generator in~\eqref{eq:L-Bohr} that
solves this problem in suitable norm.

Given a QMS with a faitful normal invariant state we may embed $\mathcal B(\mathsf h)$ into $L_2(\mathsf h)$, the space of
Hilbert-Schmidt operators on~$\mathsf h$ with inner product $\langle x,y\rangle=\operatorname{tr}(x^*y)$, in the following way:
 \[
\iota\colon \ \mathcal B(\mathsf h)\to L_2(\mathsf h), \qquad \iota(x)=\rho^{1/4}x\rho^{1/4}.
\]
Let $T=(T_t)_{t\ge 0}$ be the strongly continuous contraction semigroup on $L_2(\mathsf h)$ defined by
\[
T_t(\iota(x))=\iota(\mathcal T_t(x))
\]
and let $L$ be the generator of the semigroup $(T_t)_{t\ge 0}$. We can check that
\[
L\big(\rho^{1/4}x\rho^{1/4}\big)=\rho^{1/4}\mathcal L(x)\rho^{1/4}, \qquad \text{for }x\in \operatorname{Dom}(\mathcal{L}).
\]
The Dirichlet form, defined for $\xi\in \operatorname{Dom}(L)$, is the quadratic form $\mathcal E$
\[
\mathcal E(\xi)=-\operatorname{Re}\langle\xi,L(\xi)\rangle.
\]
The spectral gap of the operator $L$ is the nonnegative number
\[
{\rm gap}(L):=\inf\big\{\mathcal E(\xi)\,|\,\|\xi\|=1,\, \xi\in ({\rm Ker}(L))^\perp\big\}.
\]
Since rank-one operators $|e_j\rangle\langle e_k|$ $(j\not=k$) are eigenvectors for $\mathcal{L}$ (as for
all generic QMSs~\cite{FV}), and the diagonal algebra is invariant we have the same
properties also for the induced semigroup~$T$ on~$L_2(\mathsf{h})$. Let~$\mathcal{D}$ be diagonal
operators $\sum_{n\geq 0} \xi_n |e_n\rangle\langle e_n|$ in $L_2(\mathsf{h})$, i.e., such that
$\sum_{n\geq 0} |\xi_n|^2<+\infty$ and let $\mathcal{D}_{jk}$ be the linear space generated
by $|e_j\rangle\langle e_k|\, j,k\geq 0$, $j\not=k$.
One can easily check that
\[
L_2(\mathsf{h})=\mathcal{D}\oplus_{j,k\geq 0,\, j\not=k}\mathcal{D}_{jk}
\]
and, for all $\xi=\sum_{j,k\geq 0} \xi_{jk} |e_j\rangle\langle e_k|$ in $\operatorname{Dom}(L)$,
defining $\xi_0 =\sum_{j\geq 0}|\xi_{jj}|^2 e_j\rangle\langle e_j|$, it turns out that
\[
\mathcal{E}(\xi) = \mathcal{E}(\xi_0) + \sum_{j,k\geq 0,\, j\not=k}|\xi_{jk}|^2\,\mathcal{E}(|e_j\rangle\langle e_k|).
\]
As a consequence we have the following
\begin{Proposition}\label{prop:gap_computation}
Let $g_0= \inf\big\{ \mathcal E(\xi)\,|\,\|\xi\|=1,\, \xi\in \mathcal{D},\, \xi\perp\rho^{1/2} \big\}$. The spectral gap of $L$ is
\[
\operatorname{gap}(L)=\min \big\{ g_0, \, \inf_{j,k\geq 0,\,j\not=k} \mathcal{E}(|e_j\rangle\langle e_k|) \big\}.
\]
\end{Proposition}

We will now study separately the off-diagonal minimum and the diagonal spectral gap beginning by
the former that we can compute explicitly.

\subsection{The off-diagonal minimum}
Note that, for $j\not=k$, $\mathcal{L}(|e_j\rangle\langle e_k|) $ is
given by (\ref{eq:L-off-diag}) and the action of the generator $L$ of the semigroup in $L_2(\mathsf{h})$
of the invariant state $L\big(\rho^{1/4}x \rho^{1/4}\big) = \rho^{1/4}\mathcal{L}(x) \rho^{1/4}$ is the same by
\[
L(|e_j\rangle\langle e_k|)=\rho^{1/4}\mathcal{L}\big(\rho^{-1/4}|e_j\rangle\langle e_k|\rho^{-1/4}\big) \rho^{1/4}
=\mathcal{L}(|e_j\rangle\langle e_k|).
\]
Therefore, by (\ref{eq:L-off-diag}), it suffices to find the minimum on $j\not=k$ of
\begin{gather}
 \frac{1}{2}\big(\Gamma^{-}_j \varepsilon_j+ \Gamma^{-}_k \varepsilon_k
+ \Gamma^{+}_{j+1} \varepsilon_{j+1}+ \Gamma^{+}_{k+1} \varepsilon_{k+1}\big) \nonumber \\
\qquad{} = \frac{1}{2}\left(\frac{{\rm e}^{\beta \omega_j}\,\varepsilon_j}{{\rm e}^{\beta \omega_j}-1}
+ \frac{{\rm e}^{\beta \omega_k}\,\varepsilon_k}{{\rm e}^{\beta \omega_k}-1}
+\frac{ \varepsilon_{j+1}}{{\rm e}^{\beta \omega_{j+1}}-1}
+ \frac{ \varepsilon_{k+1}}{{\rm e}^{\beta \omega_{k+1}}-1} \right).\label{eq:jk-eigenv}
\end{gather}

As a result, we can compute the off-diagonal minimum after the following preliminary

\begin{Lemma}\label{lem:omega-eps}
The following hold:
\begin{enumerate}\itemsep=0pt
\item[$1.$]
The sequence $(\omega_k)_{k\geq 1}$ is non-decreasing if and only if $r+q\geq 2$.
\item[$2.$] If $r+q\geq 2$, for any $c\geq 0$ we have
\begin{equation}\label{eq:ek-div-omegak}
\frac{\varepsilon_{k}}{\omega_k}\geq 1+\frac{1}{r+q+c}
\end{equation}
for all $k\geq 2$ if and only if $(1+c)(r+q)+rq\geq 0$. In particular, if we fix $c=1$, the
inequality~\eqref{eq:ek-div-omegak} holds for all $r>1$ and $-2/3\leq q \leq 1$.
\end{enumerate}
\end{Lemma}

\begin{proof} 1. Write $\omega_k =\left(r^{k-1}(r-1)+q^{k-1}(1-q)\right)/(r-q)$
and note that, for all $k\geq 1$,
\begin{align*}
(r-q)(\omega_{k+1}-\omega_k) & = r^{k-1}(r-1)^2-q^{k-1}(1-q)^2\\
& \geq (r-1)^2 - (1-q)^2 = (r-q)(r+q-2).
\end{align*}

2. The claimed inequality is equivalent to $\omega_{k}/\varepsilon_{k}\leq 1- 1/(r+q+c+1)$ which is,
in turn, equivalent to
\begin{equation}\label{eq:ek-div-ek-1}
\frac{\varepsilon_k}{\varepsilon_{k-1}} = \frac{r^k-q^k}{r^{k-1}-q^{k-1}} \leq r+q+1+c.
\end{equation}
We show that it is equivalent to $(1+c)(r+q)+rq\geq 0$ distinguishing two cases according to the sign of $q$.

If $0\leq q\le 1$, then the sequence $(\varepsilon_k/\varepsilon_{k-1})_{k\geq 2}$ is non-increasing. Indeed
defining $f\colon [2,+\infty[{} \to {}]0,+\infty[$ by
\[
f(x) = \big(r^x-q^x\big)/\big(r^{x-1}-q^{x-1}\big)
\]
we have
\begin{align*}
f'(x) & = \frac{ (r^x\log(r)-q^x\log(q) )\big(r^{x-1}-q^{x-1}\big)
-\big(r^{x-1}\log(r)-q^{x-1}\log(q)\big) (r^x-q^x )}
{\left(r^{x-1}-q^{x-1}\right)^{2}}\\
& = \frac{-(rq)^{x-1}(r-q) (\log(r)-\log(q) )}{\big(r^{x-1}-q^{x-1}\big)^{2}} \leq 0.
\end{align*}
Hence $\sup\limits_{k\geq 2}\varepsilon_{k}/\varepsilon_{k-1}=\varepsilon_{2}/\varepsilon_{1} = r+q$ and~(\ref{eq:ek-div-ek-1}) is obviously true as well as $(1+c)(r+q)+rq\geq 0$.

If $-1\leq q <0$, then considering $g(x)=\big(r^{2x+1}+|q|^{2x+1}\big)/\big(r^{2x}-|q|^{2x}\big)$ and $h(x)=\big(r^{2x}-|q|^{2x}\big)/\big(r^{2x-1}+|q|^{2x-1}\big)$ instead of $f(x)$, we immediately get
\[
\sup_{k\geq 1}\frac{\varepsilon_{2k}}{\varepsilon_{2k-1}} =\lim_{k\to\infty}\frac{\varepsilon_{2k}}{\varepsilon_{2k-1}} =r,
\qquad
\sup_{k\geq 1}\frac{\varepsilon_{2k+1}}{\varepsilon_{2k}} = \frac{\varepsilon_3}{\varepsilon_2}
=\frac{r^2+rq+q^2}{r+q}.
\]
We easily see that $r\leq \big(r^2+rq+q^2\big)/(r+q)$ and so, in the case $-1\leq q < 0$, the supremum of
$\varepsilon_{k}/\varepsilon_{k-1}$ for $k\geq 2$ is smaller than $r+q+1+c$ if and only if
$\big(r^2+rq+q^2\big)/(r+q)\leq r+q+1+c$ which is equivalent to $(1+c)(r+q)+rq\geq 0$.

Finally, if we fix $c=1$, then $(1+c)(r+q)+rq\geq 0$ if and only if $q\geq -2r/(2+r)$, and from $r\geq 1$,
we find $q\geq -2/3$.
\end{proof}

\begin{Theorem}\label{thm:off-diag_gap}
For $-1\leq q \leq 1<r$ such that $r+q>1$ there is a pair $(j_0,k_0)$ $(j_0\not=k_0)$ such that
${\rm gap}_{\text{\rm off-diag}}(L)=\mathcal{E}(|e_{j_0}\rangle\langle e_{k_0}|)$. In particular, if $r+q\geq 2$,
$-2/3\leq q \leq 1$ then the pair is $(0,1)$ or $(1,0)$ and the off-diagonal minimum is given by
\[
{\rm gap}_{\text{\rm off-diag}}(L)
=\frac{1}{2} \left(\frac{{\rm e}^\beta+1}{{\rm e}^{\beta}-1}
+\frac{r+q}{{\rm e}^{\beta (r+q-1)}-1}\right).
\]
\end{Theorem}

\begin{proof} The first claim is an immediate consequence of
$\lim\limits_{j\to +\infty} \varepsilon_j = +\infty$ and $\lim\limits_{j\to +\infty} \omega_j = +\infty$
so that the first two terms in (\ref{eq:jk-eigenv}) diverge as $j$ and $k$ go to infinity.

Suppose now $r+q\geq 2$. In order to find the minimum for $j\not=k$ of (\ref{eq:jk-eigenv}) we first note that, for $j=1,k=2$ we have
\begin{gather*}
\frac{{\rm e}^{\beta \omega_j} \varepsilon_j}{{\rm e}^{\beta \omega_j}-1}
+ \frac{{\rm e}^{\beta \omega_k} \varepsilon_k}{{\rm e}^{\beta \omega_k}-1}
+\frac{ \varepsilon_{j+1}}{{\rm e}^{\beta \omega_{j+1}}-1}
+ \frac{ \varepsilon_{k+1}}{{\rm e}^{\beta \omega_{k+1}}-1}\Big|_{j=1,k=2} \\
\qquad{} = \frac{{\rm e}^{\beta}}{{\rm e}^{\beta}-1}
+ \frac{ {\rm e}^{\beta \omega_2} \varepsilon_2}{{\rm e}^{\beta \omega_2}-1}
+\frac{ \varepsilon_{2}}{{\rm e}^{\beta \omega_{2}}-1}
+ \frac{ \varepsilon_{3}}{{\rm e}^{\beta \omega_{3}}-1} \\
\qquad{} = \frac{{\rm e}^{\beta}}{{\rm e}^{\beta}-1}
+ \frac{ {\rm e}^{\beta (r+q-1)}\,(r+q)}{{\rm e}^{\beta (r+q-1)}-1}
+\frac{ r+q}{{\rm e}^{\beta (r+q-1)}-1}
+ \frac{ \varepsilon_{3}}{{\rm e}^{\beta \omega_{3}}-1} \\
\qquad{} \geq \frac{{\rm e}^{\beta}}{{\rm e}^{\beta}-1}
+ \frac{ {\rm e}^{\beta (r+q-1)} (r+q)}{{\rm e}^{\beta (r+q-1)}-1}
+\frac{ r+q}{{\rm e}^{\beta (r+q-1)}-1}.
\end{gather*}
The right-hand side will be bigger or equal than
\[
\frac{{\rm e}^\beta+1}{{\rm e}^{\beta}-1}+\frac{r+q}{{\rm e}^{\beta (r+q-1)}-1}
\]
if the second term satisfies
\[
\frac{ {\rm e}^{\beta (r+q-1)} (r+q)}{{\rm e}^{\beta (r+q-1)}-1} \geq \frac{1}{{\rm e}^{\beta}-1},
\]
i.e., taking inverses, $1- {\rm e}^{-\beta (r+q-1)}\leq (r+q)({\rm e}^{\beta}-1)$
which holds true because
\[
1- {\rm e}^{-\beta (r+q-1)}\leq \beta(r+q-1) < \beta(r+q) < (r+q)\big({\rm e}^{\beta}-1\big).
\]

If $2\leq j < k$, first recall that the sequence $(\omega_k)_{k\geq 1}$ is non-decreasing by Lemma
\ref{lem:omega-eps}(1). for $r+q\geq 2$. Then note that functions on $]0,+\infty[$
\begin{equation}\label{eq:monotonicity}
 x\mapsto\frac{x\, {\rm e}^{\beta x}}{{\rm e}^{\beta x}-1}, \qquad
x\mapsto\frac{x\big({\rm e}^{\beta x}+1\big)}{{\rm e}^{\beta x}-1}
\end{equation}
are increasing because
\[
\frac{{\rm d}}{{\rm d}x}\frac{x {\rm e}^{\beta x}}{{\rm e}^{\beta x}-1}
= \frac{{\rm e}^{\beta x}\big({\rm e}^{\beta x}- 1 - \beta x\big)}{\big({\rm e}^{\beta x}-1\big)^2} \geq 0,
\qquad
\frac{{\rm d}}{{\rm d}x}\frac{x\big({\rm e}^{\beta x}+1\big)}{{\rm e}^{\beta x}-1}
=\frac{\sinh(\beta x)-\beta x}{2(\sinh(\beta x))^2}\geq 0
\]
by the elementary inequalities ${\rm e}^{\beta x}\geq 1+\beta x$ and $\sinh(\beta x)\geq \beta x$.
Therefore we have the inequality
\begin{gather*}
\frac{{\rm e}^{\beta \omega_j}\,\varepsilon_j}{{\rm e}^{\beta \omega_j}-1}
+\frac{\varepsilon_{j+1}}{{\rm e}^{\beta \omega_{j+1}}-1}
+\frac{{\rm e}^{\beta \omega_k}\,\varepsilon_k}{{\rm e}^{\beta \omega_k}-1}
+\frac{\varepsilon_{k+1}}{{\rm e}^{\beta \omega_{k+1}}-1}
\geq \frac{\varepsilon_j}{\omega_j}\frac{{\rm e}^{\beta \omega_j}\,\omega_j}{{\rm e}^{\beta \omega_j}-1}
+\frac{\varepsilon_{j+1}}{{\rm e}^{\beta \omega_{j+1}}-1}
+\frac{\varepsilon_k}{\omega_k}\frac{{\rm e}^{\beta \omega_k}\,\omega_k}{{\rm e}^{\beta \omega_k}-1}
\end{gather*}
dropping the last term in the left-hand side, multiplying and dividing the first (resp.\ third) by~$\omega_j$
(resp.~$\omega_k$). Now, by monotonicity of~(\ref{eq:monotonicity}), multiplying and dividing the second term
in the right-hand side by $\omega_{j+1}$, we find
\begin{gather*}
\frac{{\rm e}^{\beta \omega_j}\varepsilon_j}{{\rm e}^{\beta \omega_j}-1}
+\frac{\varepsilon_{j+1}}{{\rm e}^{\beta \omega_{j+1}}-1}
+\frac{{\rm e}^{\beta \omega_k}\varepsilon_k}{{\rm e}^{\beta \omega_k}-1}
+\frac{\varepsilon_{k+1}}{{\rm e}^{\beta \omega_{k+1}}-1}\\
\qquad{}
\geq \frac{\varepsilon_j}{\omega_j}\frac{{\rm e}^{\beta \omega_j}\omega_j}{{\rm e}^{\beta \omega_j}-1}
+\frac{\varepsilon_{j+1}}{\omega_{j+1}}\frac{\omega_{j+1}}{{\rm e}^{\beta \omega_{j+1}}-1}
+\frac{\varepsilon_k}{\omega_k}\frac{{\rm e}^{\beta \omega_{j+1}}\omega_{j+1}}{{\rm e}^{\beta \omega_{j+1}}-1}.
\end{gather*}
Finally, by Lemma \ref{lem:omega-eps}(2). and monotonicity of (\ref{eq:monotonicity})
\begin{gather*}
\frac{{\rm e}^{\beta \omega_j}\varepsilon_j}{{\rm e}^{\beta \omega_j}-1}
+\frac{\varepsilon_{j+1}}{{\rm e}^{\beta \omega_{j+1}}-1}
+\frac{{\rm e}^{\beta \omega_k}\varepsilon_k}{{\rm e}^{\beta \omega_k}-1}
+\frac{\varepsilon_{k+1}}{{\rm e}^{\beta \omega_{k+1}}-1} \\
\qquad{}\geq
\frac{r+q+2}{r+q+1}\left(\frac{{\rm e}^{\beta \omega_j}\omega_j}{{\rm e}^{\beta \omega_j}-1}
+\frac{\left({\rm e}^{\beta \omega_{j+1}}+1\right)\omega_{j+1}}{{\rm e}^{\beta \omega_{j+1}}-1}\right)
 \geq \frac{r+q+2}{r+q+1}\left(\frac{{\rm e}^{\beta\omega_2}\omega_2}{{\rm e}^{\beta \omega_2}-1}
+\frac{{\rm e}^{\beta }+1}{{\rm e}^{\beta}-1}\right).
\end{gather*}
The difference of the right-hand side and twice the claimed lower bound is
\begin{gather*}
 \frac{\omega_2 +3 }{\omega_2+2}\left(\frac{{\rm e}^{\beta\omega_2}\,\omega_2}{{\rm e}^{\beta \omega_2}-1}
+\frac{{\rm e}^{\beta }+1}{{\rm e}^{\beta}-1}\right)
- \left(\frac{\omega_2+1}{{\rm e}^{\beta \omega_2}-1}
+\frac{{\rm e}^{\beta }+1}{{\rm e}^{\beta}-1}\right)\\
\qquad{} = \frac{1}{\omega_2+2}\left(\frac{{\rm e}^{\beta }+1}{{\rm e}^{\beta}-1}
+\frac{{\rm e}^{\beta\omega_2}\,\omega_2(\omega_2+3)-(\omega_2+1)(\omega_2+2)}{{\rm e}^{\beta \omega_2}-1}\right)\\
\qquad {}\geq \frac{1}{\omega_2+2}\left(\frac{{\rm e}^{\beta }+1}{{\rm e}^{\beta}-1}
+\frac{\omega_2(\omega_2+3)-(\omega_2+1)(\omega_2+2)}{{\rm e}^{\beta \omega_2}-1}\right)\\
\qquad{} = \frac{1}{\omega_2+2}\left(1+\frac{2}{{\rm e}^{\beta}-1}
-\frac{2}{{\rm e}^{\beta \omega_2}-1}\right)
\geq \frac{1}{\omega_2+2}\left(1+\frac{2}{{\rm e}^{\beta}-1}
-\frac{2}{{\rm e}^{\beta}-1}\right)
= \frac{1}{\omega_2+2}\geq 0.
\end{gather*}
This shows the desired inequality.
\end{proof}

\begin{Remark}\label{rem:j=0&k=1!}
 It is worth noticing that the lower bound $-2/3$ on $q$ can be relaxed to $q\geq q_0$
for some $q_0\in[-1,-2/3[$ at the price of a stronger lower bound the other parameter $r>r_0 $ for some $r_0\in{}]1,2]$.
It suffices to consider $r_0$ such that $q_0 = -2r_0/(2+r_0)$ so that, in Lemma~\ref{lem:omega-eps} for $c=1$
we have $2(r+q)+rq\geq 0$.

Moreover, for $\beta$ small, it is not difficult to find values $r>1$ and $q$ near $-1$ for which the
off-diagonal minimum is attained at some $(j_0,j_0 +1)$ with $j_0>0$.
\end{Remark}

In the following subsections we separately investigate the diagonal spectral gap for different regions of the parameters.

\subsection{Diagonal spectral gap}
As in the analysis of off-diagonal minima, parameters move in the region $-1\le q\le 1<r$ with a restriction $r+q-2\ge 0$ so that the sequence $(\omega_n)_{n\ge 1}$ is monotone increasing.
\subsubsection{Lower bound}
We already noted that, when restricted to the diagonal subalgebra our QMS reduces to the Markov semigroup of a classical
birth and death process with birth rates $(\lambda_n)_{n\ge 0}$ and death rates $(\mu_n)_{n\ge 1}$ given respectively
by (\ref{eq:lambda-mu}), namely
\begin{equation*} 
\lambda_n=\frac{1}{{\rm e}^{\beta \omega_{n+1}}-1} \frac{r^{n+1}-q^{n+1}}{r-q}, \qquad
\mu_n=\frac{{\rm e}^{\beta \omega_{n}}}{{\rm e}^{\beta \omega_n}-1}\frac{r^n-q^n}{r-q}.
\end{equation*}
In detail, let $\jmath$ be a map defined on the subspace of $L^2(\mathsf h)$ consisting of the images of the diagonal elements under the embedding $\iota$ into the sequence space defined as follows: for each $x=\sum_{n\ge 0}x_n|e_n\rangle\langle e_n|\in \mathcal B(\mathsf h)$,
\[
\jmath\colon \ \iota(x)\mapsto \xi=(x_n)_{n\ge 0}\in \ell^2(\tilde \pi),
\]
where $\tilde\pi$ denotes the probability density of the invariant measure of the aforementioned birth and death process,
i.e., $\tilde\pi_n=\pi_n/Z_\beta={\rm e}^{-\beta\varepsilon_n}/Z_\beta$.
We easily check that $\jmath$ is a unitary isomorphism, namely $\|\iota(x)\|_{L^2(\mathsf h)}=\|\xi\|_{\ell^2(\tilde \pi)}$.
Let $\mathcal A$ be the generator of the classical birth and death process with birth rates $(\lambda_n)$ and death rates
$(\mu_n)$ defined by
\[
(\mathcal Af)_n= \begin{cases}\mu_n(f_{n-1}-f_n)+\lambda_n(f_{n+1}-f_n),&\text{for }n\ge 1,\\
\lambda_0(f_1-f_0),&\text{for }n=0 \end{cases}
\]
for $f=(f_n)_{n\ge 0}\in \ell^2(\tilde\pi)$.
For each diagonal element $x=\sum_{n\ge 0}x_n|e_n\rangle\langle e_n|\in \mathcal B(\mathsf h)$,
\[
\jmath\circ L(\iota(x))=\mathcal A\circ \jmath(\iota(x)).
\]
Therefore, the diagonal spectral gap of the generator of the QMS is equal to
\begin{equation*} 
\inf\left\{-\langle f,\mathcal Af\rangle\,|\,\|f\|^2=1,\,\sum_{n\ge 0}f_n\tilde\pi_n=0\right\},
\end{equation*}
where $\langle \cdot,\cdot\rangle $ and $\|\cdot\|$ denote the inner product and the induced norm of $\ell^2(\tilde\pi)$:
\[
\langle f,g\rangle =\sum_{n\ge 0}\overline f_ng_n\tilde\pi_n.
\]
It is easy to see that
\begin{equation*} 
-\langle f,\mathcal Af\rangle=\sum_{n\ge 0}\tilde\pi_n\lambda_n(f_{n+1}-f_n)^2.
\end{equation*}

In order to compute the diagonal spectral gap we adopt the method described in \cite{Ligg} as in \cite{CaFa,DFY} and proceed as follows.
or any $f\in \ell^2(\tilde \pi)$ with $\sum_nf_n\tilde\pi_n=0$, by the Schwarz inequality,
\begin{gather}
\|f\|^2=\sum_{y<x} (f_y-f_x)^2\tilde\pi_x\tilde\pi_y
\le Z_\beta^{-2}\sum_{x<y} \left(\sum_{u=x}^{y-1}(f_{u+1}-f_u)^2\right)(y-x)\pi_x\pi_y\nonumber\\
\hphantom{\|f\|^2}{}
=Z_\beta^{-2}\sum_{u>0}(f_{u+1}-f_u)^2\pi_u\lambda_u\left(\frac{\sum_{y>u}\pi_y}{\pi_u\lambda_u}\right)\sum_{x=0}^{u-1}(u-x)\pi_x\nonumber\\
\hphantom{\|f\|^2=}{}+Z_\beta^{-2}\sum_{u\ge 0}(f_{u+1}-f_u)^2\pi_u\lambda_u\left(\frac{\sum_{y>u}(y-u)\pi_y}{\pi_u\lambda_u}\right)\sum_{x=0}^u\pi_x.\label{eq:norm_bound}
\end{gather}
From the estimation \eqref{eq:norm_bound}, by using Lemmas \ref{lem:A}, \ref{lem:B}, \ref{lem:tail_bound} and recalling that
$\pi_{u+1}/\pi_u={\rm e}^{-\beta\omega_{u+1}}$, we get
\begin{gather*}
\|f\|^2 \le \sum_{u>0}(f_{u+1}-f_u)^2\tilde \pi_u\lambda_uZ_\beta^{-1}\frac{1}{\varepsilon_{u+1}}\left(\frac{u}{1-{\rm e}^{-\beta}}-\frac{{\rm e}^{-\beta}\big(1-{\rm e}^{-\beta u}\big)}{\big(1-{\rm e}^{-\beta}\big)^2}\right)\\
\hphantom{\|f\|^2 \le}{}
+\sum_{u\ge 0}(f_{u+1}-f_u)^2\tilde \pi_u\lambda_uZ_\beta^{-1}\frac{1}{\varepsilon_{u+1}}\frac{1}{1-{\rm e}^{-\beta\omega_{u+1}}} \frac{1-{\rm e}^{-\beta(u+1)}}{1-{\rm e}^{-\beta}}.
\end{gather*}
Noting that, since $r+q\geq 2$ and so the sequence $(\omega_k)_{k\geq 1}$ is non-decreasing,
\begin{equation}\label{eq:epu-superlin}
\varepsilon_{u+1}= \sum_{k=1}^{u+1}\omega_k\ge \sum_{k=1}^{u+1}\omega_1 = u+1
\end{equation}
for all $u\ge 1$ we have
\begin{gather*}
 Z_\beta^{-1}\frac{1}{\varepsilon_{u+1}}\left(\frac{u}{1-{\rm e}^{-\beta}}-\frac{{\rm e}^{-\beta}\big(1-{\rm e}^{-\beta u}\big)}{\big(1-{\rm e}^{-\beta}\big)^2}+\frac{1}{1-{\rm e}^{-\beta\omega_{u+1}}} \frac{1-{\rm e}^{-\beta(u+1)}}{1-{\rm e}^{-\beta}}\right)\\
 \qquad{} = Z_\beta^{-1}\frac{1}{\varepsilon_{u+1}}\left(\frac{u+1}{1-{\rm e}^{-\beta}}\right)
 \leq \frac{Z_\beta^{-1}}{1 -{\rm e}^{-\beta}}.
\end{gather*}
In addition, for $u=0$,
\[
Z_\beta^{-1}\frac{1}{\varepsilon_{u+1}}\frac{1}{1-{\rm e}^{-\beta\omega_{u+1}}}
\frac{1-{\rm e}^{-\beta(u+1)}}{1-{\rm e}^{-\beta}} = \frac{Z_\beta^{-1}}{1-{\rm e}^{-\beta}},
\]
and so
\[
\|f\|^2 \leq - \frac{Z_\beta^{-1}}{1-{\rm e}^{-\beta}} \langle f,\mathcal Af\rangle.
\]
Finally, from the trivial inequality $Z_\beta\geq 1+{\rm e}^{-\beta}$, get the following result.
\begin{Theorem}\label{thm:diagonal_gap}
Suppose that $-1\le q\le 1<r$ with $r+q-2\ge 0$. Then for all $\beta>0$,
\[
\operatorname{gap}(\mathcal A)\geq Z_\beta \big(1-{\rm e}^{-\beta}\big) \geq 1-{\rm e}^{-2\beta}.
\]
\end{Theorem}

It turns out that, as the parameters change, in certain region the diagonal gap dominates and in some other region the off-diagonal minimum dominates. For example, let us compare the diagonal gap and off-diagonal minimum with fixed $q=1$. If $r>1$ is sufficiently large, then the diagonal gap dominates the off-diagonal minimum (see Figure \ref{fig:gap_diag_off-diag(r)}). On the other hand, when $r>1$ is near to $1$ and $\beta>0$ is sufficiently small, then the off-diagonal minimum dominates (see Figure \ref{fig:gap_diag_off-diag(beta)}).

\begin{figure}[h]\centering
\includegraphics[width=0.7\textwidth]{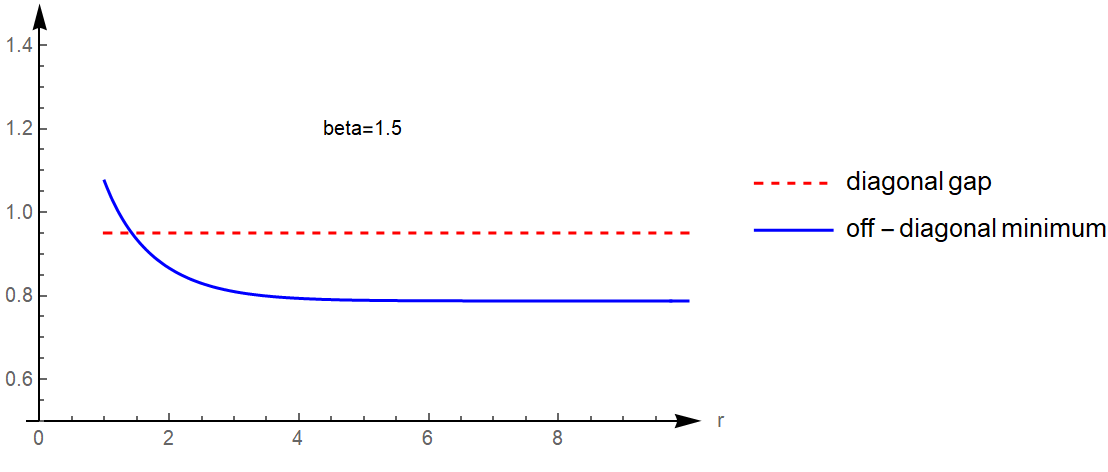}
\caption{Diagonal lower bound and off-diagonal minimum, $\beta=1.5$, $q=1$.}\label{fig:gap_diag_off-diag(r)}
\end{figure}

\begin{figure}[h]\centering
\includegraphics[width=0.7\textwidth]{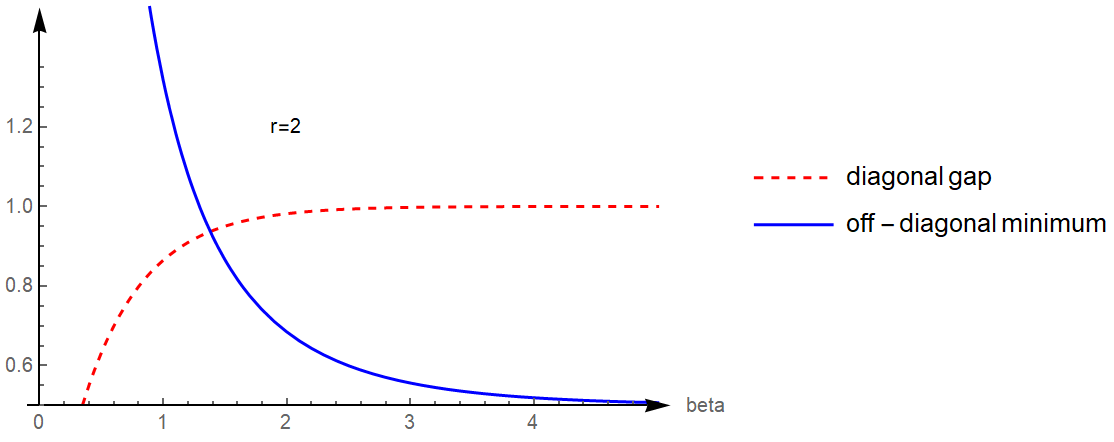}
\caption{Diagonal lower bound and off-diagonal minimum, $r=2$.}\label{fig:gap_diag_off-diag(beta)}
\end{figure}

In order to better understand which one among the off-diagonal minimum and the diagonal spectral gap is bigger
and convince ourselves that, for $\beta$ small, the spectral gap of the generator~$L$ is actually given by the
spectral gap of~$\mathcal{A}$, not just because of a poor estimate of the lower bound of Theorem~\ref{thm:diagonal_gap},
we can study the upper bound of the diagonal spectral gap.

In Appendix~\ref{sect:upp-bound}, by choosing a special $f$ and evaluating $-\langle f,\mathcal Af\rangle/\Vert f\Vert^2$, we found
\[
\operatorname{gap}(\mathcal A)\leq 1/\big(\big({\rm e}^{\beta}-1\big)\big(1-Z_\beta^{-1}\big)\big)
\]
showing that if $\beta$ is sufficiently small and $r>1$, $q<1$ are sufficiently near to $1$, then the spectral gap of the
generator $L$ coincides with the one of $\mathcal{A}$ the diagonal subalgebra. Moreover, we showed that the lower bound
of Theorem~\ref{thm:diagonal_gap} is near the optimal one for big $\beta$.

Summing up, from Proposition \ref{prop:gap_computation}, Theorems~\ref{thm:off-diag_gap}
and~\ref{thm:diagonal_gap}, we get the following
\begin{Theorem}\label{thm:sp-gap-all}
Suppose that $-1\le q\le 1<r$ with $r+q-2\ge 0$. Then for all $\beta>0$,
\[
\operatorname{gap}(L) =
\min\left\{ 1-{\rm e}^{-2\beta},\,
\frac{1}{2} \left(\frac{{\rm e}^\beta+1}{{\rm e}^{\beta}-1}
+\frac{r+q}{{\rm e}^{\beta (r+q-1)}-1}\right)\right\}.
\]
\end{Theorem}

\appendix

\section{Inequalities for partial sums}
We collect here estimates on partial sums of series needed in the evaluation of the spectral gap.
\begin{Lemma}\label{lem:A}
For all $u\in \mathbb{N}$ we have
\[
\frac{\sum_{y>u}(y-u)\pi_y}{\pi_u\lambda_u}\le \frac{1}{\varepsilon_{u+1}\big(1-{\rm e}^{-\beta \omega_{u+1}}\big)}.
\]
\end{Lemma}

\begin{proof} Recalling that the sequence $(\omega_k)_{k\geq 1}$ is non-decreasing, by $r+q-2\geq 0$, we can write
\begin{align*}
\frac{\sum_{y>u}(y-u)\pi_y}{\pi_u\lambda_u}& =\frac{{\rm e}^{\beta\omega_{u+1}}-1}{\varepsilon_{u+1}}\sum_{y>u}(y-u){\rm e}^{-\beta(\varepsilon_y-\varepsilon_u)}
=\frac{{\rm e}^{\beta\omega_{u+1}}-1}{\varepsilon_{u+1}}\sum_{y>u}(y-u){\rm e}^{-\beta\sum_{j=u+1}^y\omega_j}\\
& \le \frac{{\rm e}^{\beta\omega_{u+1}}-1}{\varepsilon_{u+1}}\sum_{y>u}(y-u){\rm e}^{-\beta(y-u)\omega_{u+1}}
=\frac{1}{\varepsilon_{u+1}\left(1-{\rm e}^{-\beta \omega_{u+1}}\right)}.\tag*{\qed}
\end{align*}
\renewcommand{\qed}{}
\end{proof}

\begin{Lemma}\label{lem:B} For all $u\in \mathbb{N}$, $u\geq 1$, we have
\begin{gather*}
\sum_{x=0}^{u-1}(u-x)\pi_x\le \frac{u}{1-{\rm e}^{-\beta}}-\frac{{\rm e}^{-\beta}\big(1-{\rm e}^{-\beta u}\big)}{\big(1-{\rm e}^{-\beta}\big)^2},\\
\sum_{x=0}^u\pi_x\le \frac{1-{\rm e}^{-\beta(u+1)}}{1-{\rm e}^{-\beta}}.
\end{gather*}
\end{Lemma}

\begin{proof}
Notice that $\varepsilon_x=\sum_{j=1}^x\omega_j\ge x\omega_1=x$ (inequality (\ref{eq:epu-superlin})). Then both follow from the explicit summation formulae.
\end{proof}

\begin{Lemma}\label{lem:tail_bound}
For the tail of the invariant measure we have the bound.
\[
\pi_{u+1} \leq \sum_{y>u}\pi_y \leq \frac{\pi_{u+1}}{1- {\rm e}^{-\beta\omega_{u+1}}}.
\]
\end{Lemma}

\begin{proof} The lower bound is obvious. For the upper bound,
\begin{align*}
\sum_{y>u}\pi_y &= \sum_{y>u}{\rm e}^{-\beta \varepsilon_y}
 = {\rm e}^{-\beta \varepsilon_u}\sum_{y>u}{\rm e}^{-\beta (\varepsilon_y-\varepsilon_u)}
 = {\rm e}^{-\beta \varepsilon_u}\sum_{y>u}{\rm e}^{-\beta \sum_{j=u+1}^y\omega_j}\\
&\le {\rm e}^{-\beta \varepsilon_u}\sum_{y>u}{\rm e}^{-\beta (y-u)\omega_{u+1}}
={\rm e}^{-\beta \varepsilon_u}\frac{{\rm e}^{-\beta\omega_{u+1}}}{1-{\rm e}^{-\beta \omega_{u+1}}}
=\frac{{\rm e}^{-\beta\varepsilon_{u+1}}}{1-{\rm e}^{-\beta \omega_{u+1}}}=\frac{\pi_{u+1}}{1-{\rm e}^{-\beta \omega_{u+1}}}.\tag*{\qed}
\end{align*}
\renewcommand{\qed}{}
\end{proof}

\section{An upper bound for the diagonal spectral gap}\label{sect:upp-bound}
We first consider the limiting case: $r\to 1^+$, $q\to 1^{-}$. In that case the jump rates become
\begin{gather*}
\lambda_n(\beta,r,q)\underset{r\to 1^+,\, q\to 1^{-}}\longrightarrow \lambda_n(\beta,1,1)=\frac{n+1}{{\rm e}^\beta-1},
\qquad \mu_n(\beta,r,q)\underset{r\to 1^+,\, q\to 1^{-}}\longrightarrow \mu_n(\beta,1,1)=\frac{n \,{\rm e}^\beta}{{\rm e}^\beta-1}
\end{gather*}
and, for small $\beta$ (i.e., high temperatures), the diagonal spectral gap is near $1$ as expected by \cite[Section~5]{CaFa}
or \cite[Section~7]{CiFaLi}.

Next, we find an upper bound with a simple $(f_n)_{n\ge 0}$.
Define for $u\ge 0$, $Z_\beta(u):=\sum_{x\ge u}\pi_x$; $Z_\beta:=Z_\beta(0)$. Put $f=(f_n)_{n\ge 0}$ as
\[
f_0=1,\qquad f_n=-c, \quad \text{for }n\ge 1,
\]
where $c$ is chosen so that $\sum_{x\ge 0}f_x\pi_x=0$, and hence $c=\pi_0/Z_\beta(1)=1/(Z_\beta -1)$,
\begin{gather*}
\mathcal E(f) = \sum_{x=0}^\infty\tilde\pi_x\lambda_x(f_{x+1}-f_x)^2 = Z_\beta^{-1}\pi_0\lambda_0(1+c)^2
= \frac{Z_\beta\,\lambda_0}{Z_\beta(1)^2}, \\
\|f\|^2 = \sum_{x=0}^\infty\tilde\pi_xf_x^2 = Z_\beta^{-1}\left(\pi_0+c^2Z_\beta(1)\right)= \frac{1}{Z_\beta(1)}.
\end{gather*}
Therefore,
\[
{\rm gap}(\mathcal A)\le \frac{Z_\beta\,\lambda_0}{Z_\beta(1)}
=\frac{1}{\big({\rm e}^{\beta}-1\big)\big(1-Z_\beta^{-1}\big)}:=\alpha(\beta,r,q).
\]
Noticing that $\alpha(\beta,r,q)$ is continuous with respect to both $\beta>0$, $r> 1$ and $0<q<1$
\[
\alpha(\beta,1,1):=\lim_{r\to 1^{+},\, q\to 1^{-}}\alpha(\beta,r,q)=\frac1{1-{\rm e}^{-\beta}}.
\]
On the other hand, comparing with the off-diagonal minimum for $r=q=1$, which is
$\big(1+3{\rm e}^{-\beta}\big)/2\big(1-{\rm e}^{-\beta}\big)$, we see that when $0<\beta<\log 3$,
\[
\frac{1}{1-{\rm e}^{-\beta}}<\frac{1+3{\rm e}^{-\beta}}{2\big(1-{\rm e}^{-\beta}\big)},
\]
which says, together with the continuity of both $\alpha(\beta,r,q)$ and the formula of off-diagonal minimum, that if $\beta$
is sufficiently small and $r>1$, $q<1$ are sufficiently near to $1$, then the spectral gap of the QMS occurs at the diagonal subalgebra.

Fix $q=1$ and let's find regions in the $r$-$\beta$ plane to see which gap would dominates. In
Figure~\ref{fig:gap dominators}, the upper line is the level curve found by equating the off-diagonal minimum and the diagonal
lower bound, $1-{\rm e}^{-2\beta}$. In the above the curve, the off-diagonal minimum is less than the lower bound of the diagonal gap,
and hence in that region the spectral gap occurs in the off-diagonal subspace. Similarly, the lower line in the figure is the level curve
obtained by equating the off-diagonal minimum and the upper bound of the diagonal gap, $\alpha(\beta,r,1)$. In the region below the curve,
the off-diagonal minimum dominates the diagonal upper bound and therefore the spectral gap occurs in the diagonal subspace.
\begin{figure}[h]\centering
\includegraphics[width=0.7\textwidth]{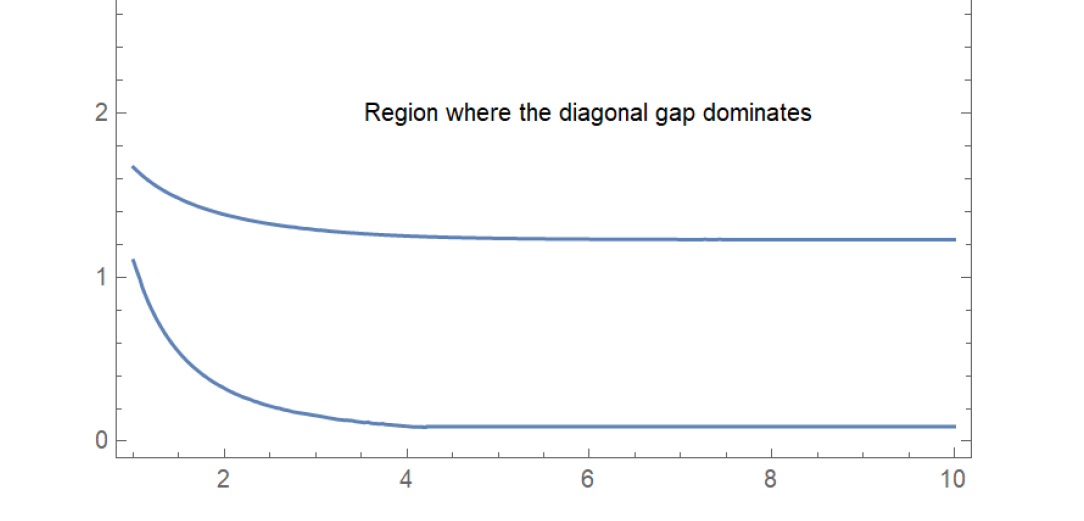}
\caption{Gap dominators: above the upper curve, the spectral gap occurs at off-diagonal subspaces and below the lower curve, the
spectral gap occurs at the diagonal subspace ($q=1$, horizontal axis for $r$ and vertical axis for $\beta$).}\label{fig:gap dominators}
\end{figure}

\subsection*{Acknowledgements}
The work of HJY was supported by the National Research Foundation of Korea (NRF) grant funded by the Korean government (MSIT) (no.\ 2020R1F1A101075). FF is a member of GNAM\-PA-INdAM Italy.
FF first met Michael Sch\"urmann at the conference Quantum Probability and Applications III held in Oberwolfach, January 25--31, 1987,
organized by their respective advisors Professors Luigi Accardi and Wilhelm von Waldenfels \cite{Sch}. Over most of these years he has had the
pleasure of meeting him at the annual QP conferences, that nowadays reached the number~42, visiting him in Greifswald, exchange views and
follow reports on his scientific work.
He would like to congratulate Michael on his retirement and wish him endless happy days with his friends and family.

\pdfbookmark[1]{References}{ref}
\LastPageEnding

\end{document}